\documentclass[10pt]{article}

\usepackage[utf8]{inputenc}
\usepackage[T1]{fontenc}
\usepackage{fullpage}

\usepackage{amssymb}
\usepackage{amsthm}
\usepackage{amsmath}
\usepackage{authblk}
\usepackage{thm-restate}
\usepackage{thmtools} 
\usepackage{enumerate}

\newtheorem{theorem}{Theorem}

\newtheorem{definition}[theorem]{Definition}

\newtheorem{conjecture}[theorem]{Conjecture}
\usepackage{amsfonts}

\DeclareMathSymbol{\leqslant}     {\mathrel}{AMSa}{"36}
\DeclareMathSymbol{\geqslant}     {\mathrel}{AMSa}{"3E}

\newcommand{\NN}{\mathbb{N}} 

\title{Global certification via perfect hashing\footnote{The authors are supported by ANR project GrR (ANR-18-CE40-0032).}}

\author[1]{Nicolas Bousquet} 
\author[1]{Laurent Feuilloley} 
\author[1]{Sébastien Zeitoun}

\affil[1]{Univ. Lyon, CNRS, Université Lyon 1, LIRIS UMR CNRS 5205, F-69621, Lyon, France}

\begin{document}
\maketitle

\begin{abstract}
	In this work, we provide an upper bound for global certification of graph homomorphism, a generalization of graph coloring. In certification, the nodes of a network should decide if the network satisfies a given property, thanks to small pieces of information called certificates. Here, there is only one global certificate which is shared by all the nodes, and the property we want to certify is the existence of a graph homomorphism to a given graph.
 
    For bipartiteness, a special case of graph homomorphism, Feuilloley and Hirvonen proved in~\cite{FeuilloleyH18} some upper and lower bounds on the size of the optimal certificate, and made the conjecture that their lower bound could be improved to match their upper bound. We prove that this conjecture is false: their lower bound was in fact optimal, and we prove it by providing the matching upper bound using a known result of perfect hashing.
\end{abstract}

\section{Introduction}
\label{sec:intro}

The topic of certification originates from self-stabilization in distributed computing, and consists in the following. Nodes of a network are provided with a unique identifier, and with some pieces of information called \emph{certificates}. These certificates can either be local (each node receive its own certificate), or global (there is a unique certificate, which is the same for all the nodes). The aim of the nodes is then to decide if the network satisfies a given property. To do so, each node should take its decision (accept or reject) based only on its local view in the network, which consists in its neighbors, their identifiers and their certificates.
The correctness requirement for a certification scheme is the following one: for every network, the property is satisfied if, and only if, there exists an assignment of the certificates such that all the nodes accept.
Unsurprisingly, the parameter we want to optimize is the size of the certificates, which is usually expressed as a function of $n$, the number of nodes in the network.
For a given property~$\mathcal{P}$, the optimal size of the certificates can be seen in some sense as a measure on the locality of~$\mathcal{P}$: the smaller it is, the more local~$\mathcal{P}$ is. 
We refer to the survey~\cite{Feuilloley21} for an introduction to certification.

As mentioned above, there are two kinds of locality in certification. In one case, the certificates are local, and the verification is local too; in the other case, the certificate is global, but the verification remains local. When speaking about local or global certification, we thus refer to the locality or globality of the certificate (and not of the verification, which is always local). In general, these two kinds of certification are somehow linked, because bounds for one can be derived from bounds for the other. Namely, a global certification scheme is a particular case of a local one, and conversely, a local certification scheme can be transformed into a global one by giving as global certificate the list of the local certificates of each node in the network (so that each node can simulate the local certification scheme by recovering its own local certificate from the global one, see~\cite{FeuilloleyH18} for more details).
However, these generic transformations are often not optimal.

In this work, the property we want to certify is the existence of a homomorphism to a given graph $H$. A particular case which has already been studied in~\cite{FeuilloleyH18} is bipartiteness (it corresponds to the case where $H$ is a clique on two vertices). Note that there exists a local certification scheme for bipartiteness using only one bit per vertex (where the certificate is the color in a proper two-coloring, and the verification of every node just consists in checking if it received a different color from all its neighbors).
Here, we focus on global certification, and with a global certificate it is less clear how to certify it. Authors in~\cite{FeuilloleyH18} made the following Conjecture~\ref{conj:n*log n} (which is also discussed in~\cite{Feuilloley21}, see Open Problem 9), in the standard case where the range of identifiers is polynomial in $n$:

\begin{conjecture}
	\label{conj:n*log n}
	The optimal size for global certification of bipartiteness is $\Theta(n \log n)$.
\end{conjecture}

In~\cite{FeuilloleyH18}, the authors proved upper and lower bounds, both parametrized by $n$ (the number of vertices in the graph), and by the range of identifiers, denoted by $M(n)$ (or simply~$M$, keeping in mind that it is a function of $n$). More precisely, they proved the following:

\begin{theorem}
	\label{thm:existing_bounds}
	Let $s$ denote the optimal size for global certification of bipartiteness. Then, we have:
	$$s = \Omega(n + \log \log M) \qquad \text{and} \qquad s = O(\min\{M, n \log M\})$$
\end{theorem}

In the standard case where $M = n^c$ for some constant $c>1$, Conjecture~\ref{conj:n*log n} is equivalent as saying that the lower bound of Theorem~\ref{thm:existing_bounds} can be improved to match the upper bound. It would also mean that the generic transformation which turns a local certification scheme of size $O(1)$ into a global one of size $O(n \log n)$ (where the global certificate is the list of the local certificates with each corresponding identifier), is optimal for bipartiteness.

In fact, we show that Conjecture~\ref{conj:n*log n} is false. Interestingly, it turns out that the lower bound of Theorem~\ref{thm:existing_bounds} is optimal, as stated in Theorem~\ref{thm:new_bound}.

\begin{theorem}
	\label{thm:new_bound}
	There exists a global certification scheme for bipartiteness with a certificate of size $O(n + \log \log M)$.
\end{theorem}

Note that, in the standard case where $M$ is polynomial in $n$, it gives a certificate of size $\Theta(n)$, which is better than the generic transformation from $O(1)$-local certificates to a $O(n \log n)$-global one, corresponding to Conjecture~\ref{conj:n*log n}. Note also that this bound remains $\Theta(n)$ even in the case where $M = 2^{2^{O(n)}}$ (while the previous upper bound provided by Theorem~\ref{thm:existing_bounds} would be $2^{O(n)}$ in that case).
\medskip

We actually prove a generalization of Theorem~\ref{thm:new_bound}, in terms of graph homomorphisms. Remember that a \emph{homomorphism} from a graph $G$ to a graph $H$ is a function $\varphi : V(G) \rightarrow V(H)$ such that, for every edge $\{u,v\} \in E(G)$, we have $\{\varphi(u), \varphi(v)\} \in E(H)$. Graph homomorphisms generalize graph colorings, since one can easily remark that a graph is $k$-colorable if and only if there exists a homomorphism from $G$ to the clique on $k$ vertices.
For example, a graph is bipartite if and only if there is a homomorphism from $G$ to an edge.

Our main result is then the following. 

\begin{restatable}{theorem}{ThmUpperBound}
	\label{thm:graph_homomorphism}
	Let $H=(V',E')$ be a graph. There exists a global certification scheme for the existence of a homomorphism to $H$ with a certificate of size $O(n \log n' + \log \log M)$ (where $n'=|V(H')|$).
\end{restatable}

Finally, let us give some intuition on the proof technique used to obtain the bound of Theorem~\ref{thm:new_bound} (which is the same as in Theorem~\ref{thm:graph_homomorphism} because it is just a particular case). As well as in the proof of the upper bounds of Theorem~\ref{thm:existing_bounds}, the prover writes a proper two-coloring in the certificate. Then, each vertex recovers its own color and the colors of its neighbors, and checks if the coloring is locally correct. What differs is the way to encode this coloring. For the $O(M)$ bound, the prover gives as certificate a list of $M$ bits, where the color of the vertex with identifier $i \in \{0, \ldots, M-1\}$ is the $i$-th bit of the list. For the $O(n \log M)$ bound, the certificate is the following: for each identifier $i$ appearing in the graph, the prover writes~$i$ (with $O(\log M)$ bits) together with the color of the vertex having the identifier $i$. In the new upper bound of Theorem~\ref{thm:new_bound}, the idea is to somehow compress the identifiers in the range $\{1, \ldots, n\}$, and then use the same technique as for the $O(M)$ bound. The compression phase is performed using a known result of perfect hashing, stated in Theorem~\ref{thm:perfect_hashing}. This result have independently been used in~\cite{EsperetHZ23} with another type of labeling, but to our knowledge, it is the first time that perfect hashing is used in distributed computing. We hope that this technique could have other applications in future works, in particular for problems related to space complexity.

\section{Model and definitions}
\label{sec:model}

For completeness, let us remind some basic graph definitions. All the graphs we consider are finite, simple, and non-oriented. Let $G=(V,E)$ be a graph. For every $u \in V$, we denote by $N(u)$ the \emph{open neighborhood} of $u$, which is set of vertices $v \in V$ such that $\{u,v\} \in E$.
A~\emph{proper two-coloring} of $G$ is a function $\varphi : V \rightarrow \{0,1\}$ such that, for every $u \in V$ and $v \in N(u)$, we have $\varphi(u) \neq \varphi(v)$. We remind that a graph $G$ is bipartite if and only if it has a proper two-coloring.

Now, let us define formally the model of certification. Let $M : \NN \rightarrow \NN$, called the \emph{identifier range} (which is fixed: it is part of the framework for which certification schemes will be designed). Let $n = |V|$. In the following, we just write $M$ instead of $M(n)$ to have lighter notations.
An \emph{identifier assignment} of $G$ is an injective mapping $Id : V \rightarrow \{0, \ldots, M-1\}$.
Finally, let $C$ be a set, called the set of \emph{certificates}.

\begin{definition}
	Let $Id$ be an identifier assignment of $G$, and $c \in C$ (called the \emph{global certificate}). Let $u \in V$. The \emph{view of $u$} consists in all the information available in its neighborhood, that is:
	\begin{itemize}
		\item its own identifier $Id(u)$;
		\item the set of identifiers of its neighbors, which is $\{Id(v) \; | \; v \in N(u)\}$;
		\item the global certificate $c$.
	\end{itemize}
\end{definition}

A \emph{verification algorithm} is a function which takes as input the view of a vertex, and outputs a decision (\emph{accept} or \emph{reject}).

Let $\mathcal{P}$ be a property on graphs. We say that there is a global certification scheme with size $s(n)$ and identifier range $M$ if there exists a verification algorithm $A$ such that, for all $n \in \NN$, there exists set $C$ of size $2^{s(n)}$ satisfying the following condition: for every graph $G$ with $n$ vertices, $G$ satisfies $\mathcal{P}$ if and only if, for every identifier assignment $Id$ with range $M$, there exists a certificate $c \in C$ such that $A$ accepts on every vertex.

	A verification algorithm is just a function, with no more requirements. In particular, it does not have to be decidable. However, in practice, when designing a certification scheme to prove upper bounds, it turns out to be decidable and often computable in polynomial time. The fact that no assumptions are made on this verification function in the definition just strengthens the results when proving lower bounds, by showing that it does not come from computational limits.

Let us give a last definition, about perfect hashing.
\begin{definition}
	\label{def:perfect_hashing}
	Let $k, \ell \in \NN$ with $k \leqslant \ell$, and let $H$ be a set of functions \mbox{$\{0, \ldots, \ell-1\} \rightarrow \{0, \ldots, k-1\}$}.
	\begin{enumerate}[a)]
		\item A function $h \in H$ is a \emph{perfect hash function} for $S \subseteq \{0, \ldots, \ell-1\}$ if $h(x) \neq h(y)$ for all $x, y \in S$, $x \neq y$.
		\item The family of functions $H$ is a $(k, \ell)$-perfect hash family if, for every $S \subseteq \{0, \ldots, \ell-1\}$ with $|S|=k$, there exists $h \in H$ which is perfect for $S$.
	\end{enumerate}
\end{definition}

\section{Main result}
\label{sec:main_result}

Let us now prove our main result:

\ThmUpperBound*

The key ingredient to prove Theorem~\ref{thm:new_bound} is the following Theorem~\ref{thm:perfect_hashing} (see e.g.~\cite{Mehlhorn84} for a proof).

\begin{theorem}
	\label{thm:perfect_hashing}
	Let $k, \ell \in \NN$ with $k \leqslant \ell$. There exists a $(k, \ell)$-perfect hash family $H_{k, \ell}$ which has size $\lceil k e^k \log \ell \rceil$.
\end{theorem}

\begin{proof}[Proof of Theorem~\ref{thm:new_bound}.]
	Let us describe a global certification scheme for the existence of a homomorphism to $H$ using a certificate of size $O(n \log n' + \log \log M)$ where $n'=|V(H)|$.
	First, since $H$ has $n'$ vertices, we can number them from $1$ to $n'$ and write the number of a vertex of $H$ on $\log n'$ bits. Similarly, for every $k, \ell \in \NN$ with $k \leqslant \ell$, by applying Theorem~\ref{thm:perfect_hashing}, we can number the functions in $H_{k, \ell}$ between $0$ and $|H_{k,\ell}|-1$. Thus, a function of $H_{k, \ell}$ can be represented using $\log |H_{k,\ell}| = O(k + \log \log \ell)$ bits.
	
    Let $G=(V,E)$ be a graph with $|V| = n$, for which there exists a homomorphism $\varphi$ from $G$ to $H$. Let $Id$ be an identifier assignment of $G$. The certificate given by the prover is the following one.
	Let us denote by $S:=\{Id(v) \; | \; v \in V\}$ the set of identifiers appearing in $G$. The set $S$ is included in $\{0, \ldots, M-1\}$ and has size $n$.
	Let $h \in H_{n, M}$ be a perfect hash function for $S$. By definition, the function $h$ induces a bijection between $S$ and $\{0, \ldots, n-1\}$. Let $L$ be the list of size $n$ such that the $i$-th element of $L$, denoted by $L[i]$, is equal to $\varphi(v)$, where $v$ is the unique vertex in $V$ such that $h(Id(v))=i$.
	The certificate given by the prover to the vertices is the triplet $(n, h, L)$, where $h$ is represented by its numbering in $H_{n, M}$. Since it uses $O(n \log n')$ bits to represent $L$ and $O(n + \log \log M)$ bits to represent $h$, the overall size of the certificate is $O(n \log n' + \log \log M)$. 
	
	Let us describe the verification algorithm. Each vertex $u$ does the following. First, it reads $n$ in the global certificate and computes $M$. Then, it can determine $h$ in $H_{n, M}$ thanks to its numbering in the certificate. Finally, $u$ accepts if and only if, for all $v \in N(u)$, $\{L[h(Id(u))],L[h(Id(v))]\} \in E'$. If it is not the case, $u$ rejects.
	
	Let us prove the correctness. First, assume that $G$ admits indeed a homomorphism to $H$. Then, by giving the certificate as described above, since $\varphi$ is a homomorphism, each vertex $u \in V$ accepts. Conversely, assume that every vertex accepts with some certificate $c$, and let us prove that there exists a homomorphism from $G$ to $H$. Since all the vertices accept, every vertex $u$ checked if $\{L[h(Id(u))],L[h(Id(v))]\} \in E'$ for every $v \in N(u)$, for some function $h$ which is written in $c$. Note that nothing ensures that $h$ is indeed a perfect hash function for the set $S$ of identifiers, but in fact, it is not necessary to check that $h$ is injective on $S$. Indeed, since every vertex $u$ accepted, then for every $v \in N(u)$, we have $\{L[h(Id(u))], L[h(Id(v))]\} \in E'$. So $\varphi(u):= L[h(Id(u))]$ defines a homomorphism from $G$ to $H$. Thus, it proves the correctness of the scheme.
\end{proof}

\section{Generalization : global certification of a Constraint Satisfaciton Problem}

More generally, perfect hashing can be used to certify the existence of a solution to a Constraint Satisfaction Problem (abbreviated into CSP). A \emph{CSP} consists in a set $V$ of variables, a domain $D$ of values for the variables, and a set $C$ of constraints. We say that it admits a solution if there is a mapping from the variables to the domain satisfying all the constraints. For instance, $k$-colorability is a particular case of a CSP, where there is one variable $x_u$ for each vertex~$u$, the domain is $\{0, \ldots, k-1\}$, and the constraints are $x_u \neq x_v$ for every edge $\{u,v\}$.

Using the same perfect hashing technique, we can design a global certification scheme in $O(n \log |D| + \log \log M)$ for the existence of a solution for any CSP with $n$ variables and domain $D$, such that the variables perform the verification, have identifiers, and each variable $v$ knows the identifiers of all the variables $w$ sharing a constraint with $v$.

\paragraph{Acknowledgments.} The authors would like to thank William Kuszmaul for fruitful discussion on hashing.

\bibliographystyle{plain}
\bibliography{biblio}

\end{document}